\documentclass{scrartcl}
\usepackage{amsfonts}
\usepackage{amsmath}
\usepackage{amssymb}
\usepackage{amsthm}
\usepackage{mathtools}
\usepackage{enumerate}
\usepackage{dsfont}
\usepackage{graphicx}
\usepackage{color}
\usepackage[all]{xy}
\usepackage{hyperref}
\usepackage{xcolor}
\usepackage{MnSymbol}
\usepackage[normalem]{ulem}
\usepackage{tikz-cd}

\usepackage{comment}
\usepackage{soul}

\newtheorem{theorem}{Theorem}
\newtheorem{proposition}{Proposition}

\newtheorem{definition}{Definition}
\newtheorem*{question*}{Question}

\newcommand{\R}{\mathbb R}

\newcommand\mc[1]{\mathcal{#1}}
\newcommand{\cB}{\mc{B}}

\newcommand{\cE}{\mc{E}}
\newcommand{\tcE}{\widetilde{\cE}}

\newcommand{\Em}{\cE_\mathrm{max}}

\newcommand{\cK}{\mc{K}}
\newcommand{\cM}{\mc{M}}
\newcommand{\cO}{\mc{O}}
\newcommand{\cP}{\mc{P}}
\newcommand{\cS}{\mc{S}}

\newcommand{\cX}{\mc{X}}

\newcommand{\cA}{\mc{A}}
\newcommand{\cH}{\mc{H}}
\newcommand{\cN}{\mc{N}}
\newcommand{\cNsa}{\mc{N}_\mathrm{sa}}
\newcommand{\cNSsa}{(\mc{N}_S)_\mathrm{sa}}

\newcommand{\BH}{\mc{B}(\cH)}
\newcommand{\BHsa}{\BH_\mathrm{sa}}

\newcommand{\PH}{{\mc{P}(\cH)}}
\newcommand{\PHone}{{\mc{P}_1(\cH)}}

\newcommand\id{{\mathrm{id}}}

\newcommand\ra{\rightarrow}

\newcommand\lra{\longrightarrow}

\begin{document}
\title{The parliament of observers}
\subtitle{Against the observable axiom of quantum mechanics}
\author{Markus Frembs\thanks{\href{mailto:markus.frembs@itp.uni-hannover.de}{markus.frembs@itp.uni-hannover.de}}, Andreas D\"oring\thanks{\href{mailto:info@andreas-doering.net}{info@andreas-doering.net}}\\
\footnotesize $^*$\textit{Institut f\"ur Theoretische Physik, Leibniz Universit\"at Hannover,}\\[-.2cm]
\footnotesize \textit{Appelstra\ss e 2, 30167 Hannover, Germany}\\
\footnotesize $^\dagger$\textit{Independent Scholar, Frankfurt am Main, Germany}}
\date{}
\maketitle
\vspace{-1cm}

\begin{abstract}
    The `unreasonable effectiveness of mathematics in the natural sciences' may at times give the impression that natural science, especially modern theoretical physics, is all but a formal apparatus. However, a key challenge of physics remains to align everyday notions of experience -- be it directly or mediated through instruments of arbitrary sophistication -- with their mathematical formalisation. This applies in particular to quantum theory, arguably the most successful theory of physics ever invented. One of its core assumptions postulates:
    \vspace{.2cm}
    \begin{center}
        \emph{(Sharp) observables correspond with self-adjoint operators.}
    \end{center}
    \vspace{.2cm}
    Here, we argue that already this most basic premise about the quantum formalism falls apart under careful scrutiny, at least when understood for all observables at once. In other words, we challenge the role of the noncommutative algebra of bounded operators as an adequate formalisation of the `algebra of observables', understood as (equivalence classes of) measurement procedures that admit a coherent empirical interpretation.
\end{abstract}

\tableofcontents

\vspace{.75cm}
\begin{center}
\textit{``Die physikalischen Größen $a$ eines gegebenen Systems $S$ entsprechen eindeutig und umkehrbar den (hermitesch-)symmetrischen Operatoren des Hilbertschen Raumes.''}
\end{center}
\vspace{-.25cm}
\hfill --- J. von Neumann \cite{vonNeumann1927c}
\vspace{.25cm}

\begin{figure}[h!]
    \centering
    \includegraphics[scale=0.22]{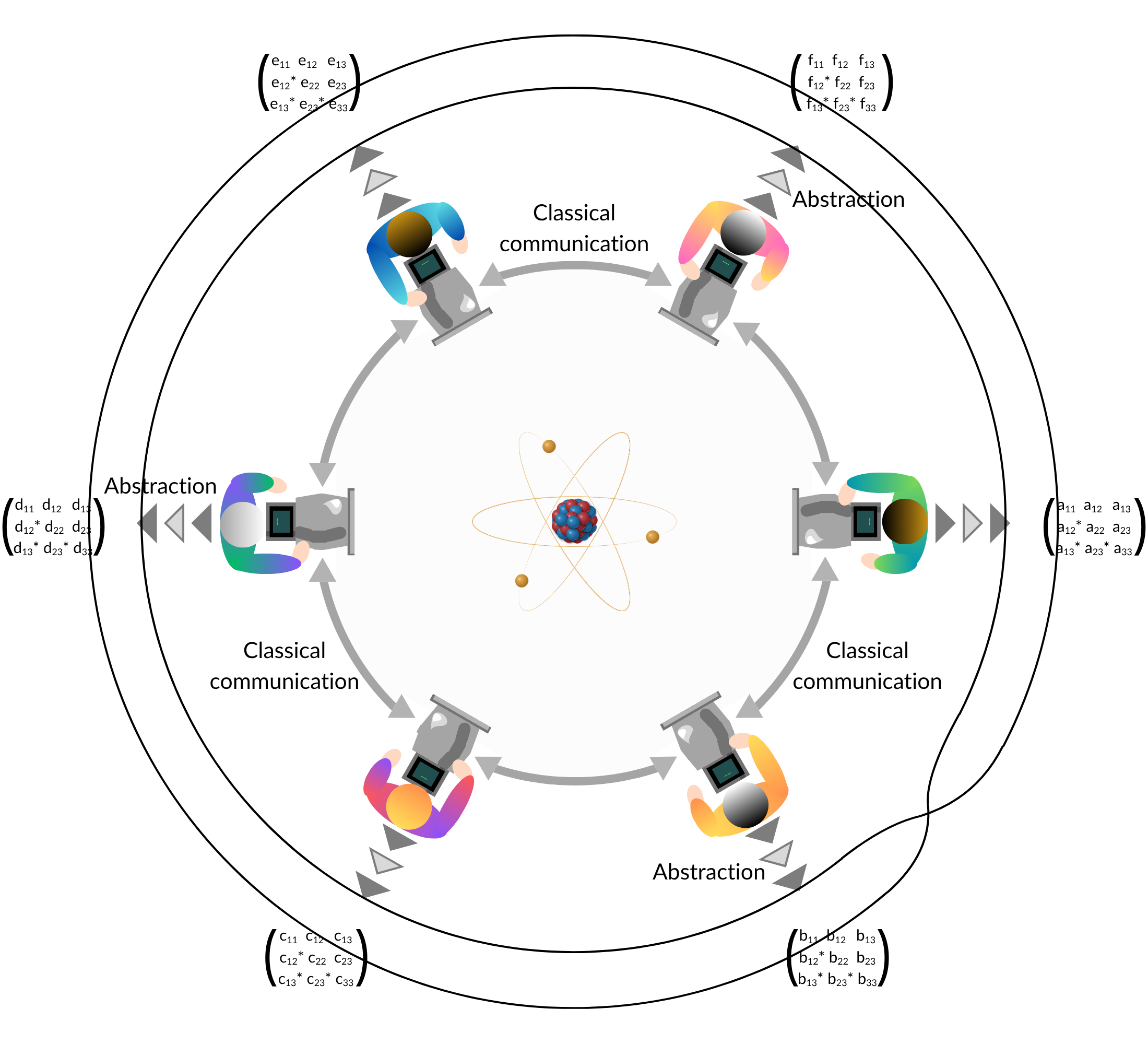}
    \caption{\footnotesize\textbf{Deconstructing the quantum-mechanical correspondence rule $(\dagger)$, `observables $\stackrel{1:1}{\longleftrightarrow}$ self-adjoint operators'.} Observers perform measurements and communicate classically among each other about their respective experimental contexts and outcomes (inner ring). As we argue in Sec.~\ref{Sec_Corr}, the usual abstraction from this data to self-adjoint operators as mathematical representatives of observables (outer ring) becomes problematic in general.}
    \label{fig: parliament}
\end{figure}

\section{The observable axiom of quantum mechanics}\label{sec: observable axiom}

After a turbulent twenty years of discoveries, mostly in atomic experiments, and a number of partially successful heuristic models, the novel quantum phenomena were found to admit a more compelling and much more broadly applicable mathematical description in terms of basic manipulations of matrices and linear operators, by Heisenberg \cite{Heisenberg1925}, Born and Jordan \cite{BornJordan1925,BornHeisenbergJordan1926}.\footnote{Schr\"odinger's parallel wave mechanics \cite{Schrodinger1926I,Schrodinger1926II,Schrodinger1926III,Schrodinger1926IV}, based on rays (`wavefunctions') in a Hilbert space describing the state of a quantum-mechanical system, was famously shown to be equivalent to Heisenberg's matrix mechanics by Schr\"odinger \cite{Schrodinger1926}, Dirac \cite{Dirac1927,Dirac1928}, Stone \cite{Stone1930} and von Neumann \cite{vonNeumann1931}.} The theory was put into its still prevalent symbolic language by Dirac \cite{Dirac1930}, and rigorously embedded within the modern Hilbert space formalism by von Neumann \cite{vonNeumann1932} (see also Ref.~\cite{Weyl1928}). This remarkable foundational work provides the authoritative basis of our understanding of the theory to this day.\footnote{With the exception of the theory of operator algebras, also pioneered by von Neumann \cite{vonNeumann1930,MurrayVonNeumann1936I,MurrayVonNeumann1937II,vonNeumann1940III,MurrayVonNeumann1943IV,vonNeumann1936}, it covers all aspects of quantum theory. Here, we will focus on the finite-dimensional case.}

Yet, a physical theory is not just mathematical formalism; the latter acquires meaning only from being (in part) a representation of experimentally accessible quantities. One of the central pillars connecting quantum theory with experimental practice is the formalisation of observables as self-adjoint operators. More precisely, to any \emph{physical system $S$} is ascribed a (complex, separable) Hilbert space $\cH_S$ and the space of bounded, self-adjoint operators $\cB_\mathrm{sa}(\cH_S)$. One of the most primitive assumptions of textbook quantum mechanics then reads as follows \cite{vonNeumann1927b,vonNeumann1927c,vonNeumann1932} (see also Ref.~\cite{Dirac1930}):
\begin{center}\label{*}
    \textbf{($*$) Observables bijectively correspond with self-adjoint operators.}
\end{center}
Of course, this statement needs some qualification.

First, from the perspective of modern measurement theory \cite{BuschLahti1996,BuschLahtiPellonpaaYlinen2016}, self-adjoint operators, equivalently projection-valued measures (PVMs), correspond with `ideal' or `sharp' observables only, a special case of more general `non-ideal' or `unsharp' observables which are identified with more general positive operator-valued measures (POVMs). This generalisation (of $(*)$) will not be relevant for our arguments, and we will therefore restrict to the ideal case corresponding with self-adjoint operators throughout.\footnote{At least formally, every POVM can also be conceptualised as a PVM on a larger system \cite{Naimark1943}.} In any case, our result is a no-go result affecting PVMs and hence also the superset of POVMs.

Second, one may riposte that $(*)$ lacks the specification of a measurement model. The standard measurement model, introduced by von Neumann \cite{vonNeumann1932}, was later shown to more naturally accommodate POVMs rather than PVMs \cite{BuschLahti1996,BuschLahtiPellonpaaYlinen2016}. In that model, measurement is specified by a quadruple $(\cK,U,Z,\phi_0)$, where $\cK$ is the Hilbert space associated with the measurement apparatus, or probe system, which is prepared in the state $\phi_0$, $U$ is a unitary interaction between probe and target system, and $Z$ is the `pointer observable' on the probe system. While measurements resulting from this model are generically POVMs, the pointer observable itself is usually taken to be a PVM, that is, a self-adjoint operator. Consequently, for the purpose of analysing the observable axiom $(*)$, introducing the interaction between system and probe explicitly (with respect to the standard model of measurement) would merely shift the focus from one system to another, and we may thus as well disregard it.

Third, the term `observable' is often used without precise definition. It is usually (implicitly) understood as referring to (equivalence classes of) `experimental setups'. In order to talk about $(*)$ with more precision, we will therefore need to spell out a minimal definition of an observable. We will do so in Sec.~\ref{sec: observables}, leaving the term as a placeholder for the time being. As a consequence of this formalisation, we will make explicit that the identification $(*)$ is to be understood not as a bare bijection, but one that also preserves post-processing capabilities of experimenters \cite{vonNeumann1927b,vonNeumann1927c}. To highlight this structural aspect, which we will spell out in more detail in Sec.~\ref{sec: breaking QM}, we will call the bijection in $(*)$ \emph{faithful}.

Fourth, with observable understood loosely as above, $(*)$ is generally too strong. There are at least two aspects to this. If the physical system exhibits symmetries or superselection rules, then not all self-adjoint operators on a given Hilbert space are associated with observables. For instance, electric charge is a classical degree of freedom, it labels different `charge sectors' which decompose the total system Hilbert space into orthogonal subspaces to which physical (pure) states are confined. Self-adjoint operators that would measure superpositions of states in different charge sectors are therefore excluded from the quantum formalism. Moreover, it is usually understood that not every self-adjoint operator corresponds with a `measurement apparatus'. For instance, the Bell operator \cite{Bell1964,CHSH1969} $B=a\otimes b+a\otimes b'+a'\otimes b-a'\otimes b'$ is a sum of tensor products of self-adjoint operators on two (spatially separated) systems $a,a'\in\cB_\mathrm{sa}(\cH_{S_A})$ and $b,b'\in\cB_\mathrm{sa}(\cH_{S_B})$ with $\cH_S=\cH_{S_A}\otimes\cH_{S_B}$, that are individually local (with respect to the tensor product decomposition), and can thus be measured by local observers; however, $B$ is best understood as describing the statistical ensemble of many runs of measuring these local operations, rather than being implementable as a single-shot measurement itself.

Our discussion will not depend on these subtleties either. The issue we point out below already arises for local systems, and we may thus restrict to systems $S$ that are local to an experimenter. We employ the usual algebraic formulation and let $\cN_S\subset\cB(\cH_S)$ be the von Neumann algebra of $S$, whose self-adjoint elements we denote by $\cNsa:=\cNSsa=\{a\in\cN_S\mid a^*=a\}$.\footnote{More generally, the \emph{quasilocal algebra} in algebraic quantum field theory is the inductive limit of the net of local algebras $\cA(O)$, where $O\subset\cM$ ranges over local regions of some (spacetime) manifold $\cM$, and which are meant to represent the observables `accessible to local observers in $O$' \cite{HaagKastler1964,Haag1996,FewsterRejzner2020}.} (From now on, we will mostly suppress the label specifying the (local) system under consideration.)

Then ($*$) assumes the following refinement (cf. Ref.~\cite{vonNeumann1927b,vonNeumann1927c,vonNeumann1932,Dirac1930}):
\begin{center}
    \textbf{($\dagger$) Sharp observables bijectively and faithfully correspond with the self-adjoint operators in $\cNsa$.}
\end{center}
We will focus on the finite-dimensional case, so we could consider hermitian matrices instead of $\cNsa$. Yet, we give the formulation with a von Neumann algebra with an eye to future generalisations.

We will refer to ($\dagger$) as the \emph{observable axiom} or \emph{observable-operator correspondence} of quantum theory. It forms a crucial part of the link between physical concepts and their mathematical representation (cf. Fig.~\ref{fig: parliament}).  With ($\dagger$), (the complex algebra $\cN$ generated by) $\cNsa$ is often called the \emph{algebra of observables of a quantum system $S$}.\footnote{The product on $\cN$ is given by composition of linear operators, which reduces to matrix multiplication in finite dimensions. Self-adjoint operators are not closed under this product, and thus generate a complex algebra. Nevertheless, $\cNsa$ is closed under, and thus forms a real algebra with respect to the symmetrised (Jordan) product (see also Sec.~\ref{sec: consequences}). This distinction will be irrelevant for our arguments, and we therefore loosely refer to $\cNsa$ as the algebra of observables.}

In this paper, we challenge $(\dagger)$ and the status of $\cNsa$ as the `observable algebra' of a quantum system. More specifically, under mild assumptions we derive a basic consistency property for a collection of observables, which we show to be incompatible with the relation $(\dagger)$, that is, with the structure of self-adjoint operators, $\cNsa$. For simplicity and conceptual clarity, we will restrict to systems with finitely many degrees of freedom.

\section{On the observable-operator correspondence}\label{Sec_Corr}

Having reviewed the observable axiom ($\dagger$) of quantum mechanics in Sec.~\ref{sec: observable axiom}, we proceed to scrutinise it in the present one. To this end, we first provide a more formal definition for the term `observable'. This requires a minimum of mathematical structure, which, as we argue, is common to any, and thus independent of a particular physical theory and their mathematical representation. In particular, it applies to classical, quantum, as well as various proposals for post-quantum theories.

\subsection{A minimal formalisation of observables}\label{sec: observables}

To formalise what is meant by an `observable', we seek a minimal, light-weight definition that does not pre-suppose anything about the embedding of these objects into the full mathematical apparatus of a physical theory. Conceptually, we take `observables' to refer to `real-world measurements',\footnote{As we will specify below, this does not require an observable to be directly `implementable'; it is sufficient that it can be post-processed in an unambiguous way from physically implementable ones.}\footnote{This immediately raises the question: a measurement of what? While it is interesting to ponder the possibility that it may not always be possible to sharply demarcate the boundary between system and observer, here we adopt the common assumption underlying all present physical theories -- including quantum theory -- that physical systems can be treated in isolation, and be given an independent mathematical description, including their observable structure. With this assumption in place, we will leave implicit the subsystem index $S$ on mathematical structures ascribed to a system.} and our formalisation is meant to extract the essential mathematical structure necessary for a description at this level.\footnote{Our formalisation of the term `observable' here is not new, but has various incarnations within logical and operational frameworks to quantum theory \cite{BirkhoffVonNeumann1936,Ludwig1953,Ludwig1964,Jauch1968,Mackey1963,Maczynski1971b,DvurecenskijKukova2014,Muller2021,FoulisRandall1972,RandallFoulis1973,PtakPulmannova1991}. \label{fn: observable}}

Crucially, we make the following premise (usually understood implicitly):
\begin{equation*}\label{ax: EA informal}
    \begin{minipage}{0.9\textwidth}
    \centering
    \itshape
    An experimenter can unambiguously refer to an observable
    
    (prior to the act of measurement, and independent of a given state);
    
    she can record her descriptions of available observables (e.g. in a laboratory notebook), and 
    communicate coherently about these with other experimenters.
    \end{minipage}
    \tag{B}
\end{equation*}
Spelling out what this means in more detail -- and why it is non-trivial! -- is the content of this and the next section.

In particular, we take for granted the possibility to disambiguate outcomes of an observable,\footnote{More generally, we will consider events, which represent statements such as `upon measurement, observable $X$ has a value in the (Borel) set $\Delta$'. Events will also be formalised below.} equivalently that it is possible to label these by elements in some alphabet, which, as is customary, we will identify with some subset of the real line $\R$. Moreover, if an observable is to refer to a `real-world measurement', it must admit a description in terms of experimentally accessible data. We will refer to such a description as an \emph{experimental context}, for now to be understood informally (we will specify experimental contexts further below). The specification of an observable thus splits as follows.

\begin{definition}[informal]\label{def: observable}
    An observable $X$ associated with a physical system $S$ is described by its experimental context $E$ and the set of its possible outcomes.\footnote{In general, the outcome set of an observable may be discrete, continuous or a mixture of these. For simplicity, here we focus on the finite case.} 
\end{definition}

To be sure, there are many details specifying a given experimental context, many of which may be irrelevant for the mathematical description of an observable. We will discuss the subtleties involved in reducing such redundancy in the next section. Before doing so, we motivate the general formal structure of experimental contexts.

Handed an observable, an experimenter (or observer) will readily start collecting data and performing statistics on it. In fact, being granted the ability to apply statistical reasoning immediately enlarges the set of observables associated with an experimenter, to all those obtained from basic algebraic manipulations and post-processing of measurement outcomes. Taking this into account, it is natural to model the experimental context $E$ (with respect to a system $S$) as a measurable space of outcomes and corresponding events $ (\Omega_E,\Sigma_E)$,\footnote{For a finite experimental context whose elementary outcomes are individually distinguishable, we take $\Sigma_E=2^{\Omega_E}$, equivalently the events $\{\omega\}$ correspond with the atoms of $\Sigma_E$.} and identify observables with bounded, measurable functions on it,
\begin{align*}
    \cA_E:=
    \cM_b(\Omega_E,\Sigma_E)
    =\{X:\Omega_E\ra\R\mid X\mathrm{\ bounded\ and\ }\Sigma_E\mathrm{-measurable}\}\; .
\end{align*}
$\cA_E$ is closed under pointwise addition, multiplication, scalar multiplication with elements of $\R$ and contains the constant function of value $1$ which acts as a unit, hence, $\cA_E$ defines a real unital commutative algebra. For $\Omega_E$ finite, this simplifies to $\cA_E=\R^{\Omega_E}$.\footnote{In addition, the supremum $||X||_\infty:=\sup_{\omega\in\Omega_E}|X(\omega)|$ defines a norm and promotes $\cA_E$ to a (unital, commutative) Banach algebra. For $\Omega_E$ finite and $\mu_E$ the counting (or any other equivalent) measure, $\cA_E$ coincides with the self-adjoint part of the von Neumann algebra $L^\infty(\Omega_E,\Sigma_E,\mu_E)$.}

\begin{definition}\label{def: experimental context}
    An experimental context $E$ of a physical system $S$ is a measurable space $(\Omega_E,\Sigma_E)$, together with the algebra of bounded, measurable functions $\cA_E=\cM_b(\Omega_E,\Sigma_E)$.
\end{definition}

Note that given the ability to post-process, the specific choice of labelling the outcome set $\Omega_E$ of a set of observables becomes `gauge'. Our main point of contention with $(\dagger)$ will be that for observables with an unambiguous reference to real-world measurements this choice of `gauge' must be consistent across experimental contexts; yet, as we will show, this is not the case for the set of self-adjoint operators $\cNsa$ associated with a quantum-mechanical system. To make this precise, we need to talk not only about a single observable, but about collections of observables and the relations between them.

Before doing so, we first observe that any bounded, measurable function $X:\Omega_X\ra\R$ defines an experimental context on its own:
\begin{align}\label{eq: X-generated AvN}
    \cA_X
    :=\cM_b(\Omega_X,\sigma(X))
    =\{g\circ X\mid g\in\cM_b(\R,\cB(\R))\}\; ,
\end{align}
where $\sigma(X)=X^{-1}(\cB(\R))=\{X^{-1}(B)\mid B\in\cB(\R)\}$ is the countably (here, finitely generated) $\sigma$-algebra associated with $X$ and $\cB(\R)$ is the Borel $\sigma$-algebra on $\R$. In the finite setting we are considering here, one simply has $\cA_X=\R^{\Omega_X}$. Since by Eq.~(\ref{eq: X-generated AvN}) every observable generates its own experimental context, we may think of a space of observables, $\cO$, equivalently as a collection of experimental contexts, $\cE$, represented by algebras $\cA_E$ for every $E\in\cE$. Our coherence condition below will be expressed from this perspective.

\subsection{Spaces of observables}\label{sec: spaces of observables}

In this subsection, we combine various observables and their experimental contexts into a single structure, which for distinction with the term `algebra of observables', which is already reserved for $\cNsa$, we will refer to simply as a \emph{space of observables $\cO$}.

Clearly, the description of a physical system will generally involve many observables and experimental contexts (cf. Def.~\ref{def: experimental context}). It is part of the mathematical structure of a specific physical theory to combine them into an economical whole, however, a priori, different experimental contexts are completely independent from one another, and thus comprise a `space of observables', which with Eq.~(\ref{eq: X-generated AvN}) is of the form\footnote{A similar construction is discussed in Ref.~\cite{KochenSpecker1967} (and elsewhere). We will address the similarities and differences with the perspective taken in Ref.~\cite{KochenSpecker1967} (and elsewhere) in Sec.~\ref{sec: KS theorem revisited} below.}
\begin{align}\label{eq: space of observables}
    \cO_0
    =\bigsqcup_{E\in\cE_0}\cA_E\; .
\end{align}
This description of a system by a disjoint union of contexts may be highly redundant, as there may be nontrivial relations among the different experimental contexts of observables that may simplify their formalisation. Clearly, the observables within a given experimental context in Def.~\ref{def: experimental context} are related via post-processing (mathematically, the Borel functional calculus). Yet, non-trivial relations may exist also between the observables of different experimental contexts, which may allow for a compressed description of $\cO_0$ in terms of a reduced space of observables $\cO$. To do so, we need to address the nature of such relations, and to what extent they can be incorporated in our description of $\cO$, while staying faithful to the premise that $\cO$ is to be understood in terms of classically communicable records.

At this point, it is useful to discuss an example. A natural equivalence relation on the space of observables and their experimental contexts is that of \emph{operational indistinguishability}. Assume we are given preparation devices, whose collective output constitutes the \emph{state space $\cS$} (of a system $S$), which generally takes the form of some convex set \cite{Barrett2007,Muller2021,Plavala2023}. Then given two observables $X,X'\in\cO_0$ with respective outcome spaces $\Omega_X,\Omega_{X'}$, we say that two outcomes (events) are \emph{operationally indistinguishable} if
\begin{align}\label{eq: operational indistinguishability}
    P(\omega\in\Omega_X\mid X;s)
    =P(\omega'\in\Omega_{X'}\mid X';s)\quad\quad\forall s\in\cS\; ,
\end{align}
where $P(\omega\in\Omega_X\mid X;s)$ denotes the probability to observe outcome $\omega$ upon performing the measurement $X$ in the state $s$. Eq.~(\ref{eq: operational indistinguishability}) defines an equivalence relation on the space of experimental contexts; given two operationally indistinguishable experimental contexts it is natural to identify them, that is, to consider the quotient $\cE_\sim=\cE_0/\sim$, equivalently the quotient $\cO_\sim=\cO_0/\sim$, where $\sim$ is the equivalence relation defined by Eq.~(\ref{eq: operational indistinguishability}).\footnote{This procedure is standard in operational approaches such as generalised probabilistic theories (GPTs) \cite{Barrett2007,Muller2021,Plavala2023}, which model a physical system by a pair $(\cS_\sim,\cO_\sim)$ with $\cS_\sim$ the convex sets of states (equivalence classes of operationally distinguishable preparation procedures).}

However, while this construction is sound mathematically, we must exercise some caution in descending to this quotient too quickly. The reason is the following: whatever relations we impose on the primitive description $\cO_0$ (e.g. Eq.~(\ref{eq: operational indistinguishability})), for the resulting space of observables $\cO$ to remain meaningful as a description of real-world measurements, it must still be possible for an experimenter (or a family of experimenters) to keep a coherent description of their respective experimental contexts as well as the relations between them. That is, an experimenter must be able to unambiguously record and communicate about the observables $X\in\cO$ and their respective outcome sets $\Omega_X$ and corresponding events $\Sigma_X$. As stated above, we take this to be a necessary, albeit implicit assumption when speaking of an `algebra of observables'. We spell this out more formally.

To say that $\cO$ admits an unambiguous description means, in particular, that both the experimental context $E$ and its outcome set $\Omega_E$ can be unambiguously described within a common language. To formalise this, consider a subset $\tcE$ of isomorphic contexts, that is, for all $E,E'\in\tcE$, we have that $(\Omega_E,\Sigma_E)$ is isomorphic to $(\Omega_{E'},\Sigma_{E'})$ as measurable spaces. Typically, one will consider a set of maximal, pairwise isomorphic contexts $\tcE$. This is what we will do in the quantum case. Now fix a reference context $(\Lambda,\Sigma_\Lambda)$ that is also isomorphic, and define for each $E\in\tcE$ a measurable isomorphism
\begin{align}\label{eq: c_E}
    c_E:(\Omega_E,\Sigma_E)\lra(\Lambda,\Sigma_\Lambda).
\end{align}
Here, $\Lambda$ is a common alphabet in which we can (re)express the outcomes of and events about observables in $\tcE$. The role of $c_E$ is to assign labels to the outcomes of observables; the particular choice of $c_E$ carries no physical significance, but is merely conventional.

Now, the important point is that in general, one will want to identify certain outcomes and events from different contexts.

Consider for the moment standard quantum theory as an example: let $X,Y$ be two maximal, i.e., non-degenerate observables that have a spectral projection $P$ in common. Then $P$ represents an event about $X$, namely `$X$ has value $x$', where $x$ is the eigenvalue of $X$ corresponding with $P$, and $P$ also represents an event about $Y$, namely `$Y$ has a value $y$', with $y$ the eigenvalue of $Y$ corresponding with $P$. Quantum theory identifies these two events, which is simply expressed by labelling both events by the common label $P$.

In any approach that aims to reconstruct quantum theory, there will necessarily be a step to make such identifications across contexts, and even more generally, in any physical theory that we are aware of, many identifications of this kind are being made, usually implicitly. The reasons may be operational, via operational indistiguishability, or theoretical, potentially aesthetical, and certainly practical. After all, the description in Eq.~(\ref{eq: space of observables}) is highly redundant and not very elegant.

We will use the notation $\Omega_E\cap\Omega_{E'}$ to indicate that certain events in the experimental contexts $E$ and $E'$ are being identified within some physical theory.\footnote{A more complete notation would be $(\Omega_E,\Sigma_E)\cap(\Omega_{E'},\Sigma_{E'})$, but this is a bit clumsy.\label{fn: overlap notation}} Note that we do not assume that our measurable spaces are already embedded into some larger space, the notation merely serves to indicate that some identifications are being made.

We now encode what it means to say that a space of observables (within a family of experimental contexts) admits a coherent description. The obvious restriction is that if a physical theory imposes identifications between experimental contexts $E$ and $E'$, that is, if $\Omega_E\cap\Omega_{E'}$ is non-empty, then the maps $c_E$ defined above should respect this:

\begin{definition}\label{def: classical coordinatisation}
    Let $\cO$ be the space of observables of a physical system, let $\tcE\subset\cE$ be a subset of pairwise isomorphic experimental contexts, and let $(\Lambda,\Sigma_\Lambda)$ be a reference measurable space, also isomorphic. A family $c=(c_E)_{E\in\tcE}$ of (measurable) isomorphisms $c_E:(\Omega_E,\Sigma_E)\ra(\Lambda,\Sigma_\Lambda)$ such that
    \begin{align}\label{eq: event identifications}
        c_E|_{\Omega_E\cap\Omega_{E'}}
        =c_{E'}|_{\Omega_E\cap\Omega_{E'}}
    \end{align}
    for all $E,E'\in\tcE$ is called a \emph{(classical) coordinatisation of $\tcE$.}
\end{definition}

A classical coordinatisation collects the various naming conventions of (the outcomes of) observables across experimental contexts in the family $\tcE$. In particular, note that Eq.~(\ref{eq: event identifications}) expresses that naming conventions of observables and their outcomes that are identified within some physical theory must coincide. There still is a lot of freedom in the choice of $c_E$: outside of $\Omega_E\cap\Omega_{E'}$, in particular, if $\Omega_E\cap\Omega_{E'}$ is empty, no condition applies and $c_E$ can be chosen arbitrarily. Yet, note that for a fixed $E\in\tcE$, the context $E'$ varies over the rest of the family $\tcE$.

For simplicity, and with (finite-dimensional) quantum systems in mind, we will henceforth assume that all maximal experimental contexts are isomorphic, and will write $\Em\subset\cE$ for the family of (pairwise isomorphic) maximal experimental contexts. Moreover, we assume that each observable is contained in at least one maximal context.

\begin{definition}\label{def: classical recordability}
    Let $\cO$ be the space of observables of a physical system, and let $\Em\subset\cE$ be the family of all (pairwise isomorphic) maximal experimental contexts. Then $\cO$ is called \emph{(consistently) recordable} if $\Em$ admits a classical coordinatisation.
\end{definition}

Def.~\ref{def: classical recordability} formalises our premise \ref{ax: EA informal}: the space of observables associated with a system is meant to admit a coherent description, that is, to be unambiguously and consistently recordable by an experimenter. Def.~\ref{def: classical recordability} should thus be read as expressing a fundamental consistency or well-posedness condition for the very notion of (a space of) observables.

Now, a classical coordinatisation $c$ induces (bijective) relations between pairs of isomorphic experimental contexts (with the same cardinality of outcomes spaces),\footnote{More generally, $\nabla_{E'E}:(\Omega_E,\Sigma_E)\to(\Omega_{E'},\Sigma_{E'})$ is a measurable isomorphism, that is, a bijection such that both $\nabla_{E'E}$ and $\nabla^{-1}_{E'E}$ are measurable. Note also that a measurable map $\nabla_{E'E}$ induces a unital algebra homomorphism of observable algebras by pullback, $\nabla^*_{E'E}:\cA_{E'}\ra\cA_E$, $\nabla^*_{E'E}(f)=f\circ\nabla_{E'E}$.}
\begin{align*}
    \nabla_{E'E}
    :=c_{E'}^{-1}\circ c_E:(\Omega_E,\Sigma_E) \lra (\Omega_{E'},\Sigma_{E'})\; .
\end{align*}
Again, no distinguished physical meaning is assigned to the maps $\nabla_{E'E}$ themselves, they merely track the relative naming conventions $c_E$ between different experimental contexts $E\in\cE$, e.g. calibration procedures, changes of experimental arrangement, or any other classically communicable relation between the respective experimental contexts. We may thus think of the maps $\nabla_{E'E}$ simply as relations between different experimental contexts recorded in a laboratory notebook.

\begin{definition}\label{def: ex context connection}
    Let $\cO$ be the space of observables of a physical system, and let $\Em\subset\cE$ be the set of (pairwise isomorphic) maximal experimental contexts. A family $\nabla=(\nabla_{E'E})_{E,E'\in\Em}$ of (measurable) isomorphisms $\nabla_{E'E}:(\Omega_E,\Sigma_E)\ra (\Omega_{E'},\Sigma_{E'})$ with $\nabla_{EE'}=\nabla^{-1}_{E'E}$ and
    \begin{align}\label{eq: overlap identifications}
        \nabla_{E'E}|_{\Omega_E\cap\Omega_{E'}}
        =\mathrm{id}_{\Omega_E\cap\Omega_{E'}}\; 
    \end{align}
    for all $E,E'\in\Em$ is called an \emph{(experimental) context connection on $\Em$}.
\end{definition}

Here, the notation $\Omega_E\cap\Omega_{E'}$ again merely indicates the identifications that one has chosen to make, for operational or theoretical reasons.$^{\ref{fn: overlap notation}}$ We do not assume that outcomes and events are already embedded in some common classical space. Note that Eq.~(\ref{eq: overlap identifications}) simply expresses that naming conventions of the outcomes (and events) of observables that are identified by some physical theory must coincide.

Now, with respect to Def.~\ref{def: ex context connection}, we can re-express our informal condition in \ref{ax: EA informal}, that is, the possibility of an experimenter to unambiguously refer to observables as follows.

\begin{proposition}\label{prop: classical coherence}
    Let $\cO$ be a recordable space of observables with corresponding set of experimental contexts $\cE$, and $\Em\subset\cE$ the subset of all (pairwise isomorphic) maximal experimental contexts. Then there exists an experimental context connection $\nabla$ on $\Em$ such that for any tuple of isomorphic contexts $\gamma=(E_0,\cdots,E_{n-1},E_n=E_0)$ in $\Em$,
    \begin{align}\label{eq: flatness}
        P_\gamma(\nabla)
        :=\circ_{i=0}^{n-1}\nabla_{E_{i+1}E_i}
        =\mathrm{id}_{\Omega_{E_0}}\; .
    \end{align}
\end{proposition}

\begin{proof}
    Since $\cO$ is recordable, there exists a classical coordinatisation $c$ for $\Em$. Let $\nabla_{E'E}:=c^{-1}_{E'}\circ c_E$ for all $E,E'\in\Em$. Clearly, we have $\nabla_{EE'}=\nabla^{-1}_{E'E}$ and Eq.~(\ref{eq: overlap identifications}) follows from Eq.~(\ref{eq: event identifications}) since $\nabla_{E'E}(\omega)=c^{-1}_{E'}\circ c_E(\omega)=\omega$ for any event $\omega\in\Omega_E,\Omega_{E'}$. Hence, $\nabla=(\nabla_{E'E})_{E,E'\in\Em}$ defines an experimental context connection on $\Em$ by Def.~\ref{def: ex context connection}. Moreover, for every tuple of experimental contexts $\gamma=(E_0,E_1,\ldots,E_{n-1},E_n=E_0)$ in $\Em$, we have
    \begin{equation*}
        P_\gamma(\nabla)
        =\circ_{i=0}^{n-1}\nabla_{E_{i+1}E_i}
        =(c_{E_0}^{-1}c_{E_{n-1}})\cdots(c_{E_2}^{-1}c_{E_1})(c_{E_1}^{-1}c_{E_0})
        =c_{E_0}^{-1}c_{E_0}
        =\id_{\Omega_{E_0}}\; .\qedhere
    \end{equation*}
\end{proof}

Note also that fixing a classical coordinatisation of some reference experimental context and transporting it along an experimental context connection satisfying Eq.~(\ref{eq: flatness}) yields a globally consistent family of classical coordinatisations. Therefore, classical coordinatisations and experimental context connections satisfying Eq.~(\ref{eq: flatness}) both express the same consistency condition, so Prop.~\ref{prop: classical coherence} formalises the informal premise \ref{ax: EA informal}, which we regard as an essential prerequisite for the formulation of any physical theory.

It is easy to see that Prop.~\ref{prop: classical coherence} is satisfied on the level of $\cO_0$ and $\cE_0$ in Eq.~(\ref{eq: space of observables}), that is, prior to any more refined formalisation within some explicitly physical theory. In particular, note that up to now we have not yet used the observable-operator correspondence $(\dagger)$ of quantum theory. We will do so in the next section, and find that $(\dagger)$ is incompatible with the notion of consistent recordability in Def.~\ref{def: classical recordability}.

\subsection{Breaking of the observable-operator correspondence}\label{sec: breaking QM}

We are now in the position to evaluate $(\dagger)$. It bijectively relates observables, that is, bounded, measurable functions $X:\Omega_X\ra\R$, with self-adjoint operators $a_X\in\cNsa$. However, as noted in Sec.~\ref{sec: observable axiom}, $(\dagger)$ is to be understood not only as a bare bijection: it should also preserve the post-processing capabilities of experimenters. More formally, $(\dagger)$ is required (see Ref.~\cite{vonNeumann1927b,vonNeumann1927c,vonNeumann1932}) to preserve functional composition, that is, if $X\stackrel{\dagger}{\ra}a_X$, then, for every admissible measurable relabelling $f$ of (the outcomes of) $X$,\footnote{Eq.~(\ref{eq: FUNC}) is usually called the \emph{functional composition principle} in discussions of the Kochen-Specker theorem (see also Sec.~\ref{sec: KS theorem revisited}).}
\begin{align}\label{eq: FUNC}
    f\circ X
    \stackrel{\dagger}{\lra}
    f(a_X)\; .
\end{align}
The reason is that $f\circ X$ need not be taken as an independent measurement procedure, but as the measurement of $X$ followed by classical post-processing of its outcomes. In turn, a representation of observables which failed to preserve this relation would be a bijection of sets, but not a faithful representation of their operational structure. We therefore take Eq.~(\ref{eq: FUNC}) to be implicitly understood as part of the observable axiom $(\dagger)$.

It follows that experimental context data is no longer retained under $(\dagger)$. More precisely, comparing with Sec.~\ref{sec: observables} and noting that every self-adjoint operator generates a commutative subalgebra $C_a=(\{1,a\}'')_\mathrm{sa}$ in $\cNsa$ (via the Borel functional calculus \cite{vonNeumann1930}),\footnote{Recall that the commutant of $\cX\subset\BH$ is defined by $\cX'=\{a\in\BH\mid[a,b]=0\ \forall b\in\cX\}$. The double commutant $C_a=\{1,a\}''$ of an operator $a\in\BHsa$ and the unit element $1$ in $\BH$ defines a commutative von Neumann algebra. In finite dimensions, the algebra $C_a$ consists of all real-linear combinations of the projections onto the eigenspaces of the self-adjoint operator $a$. Conversely, every commutative von Neumann algebra $L^\infty(\Omega,\Sigma,\mu)$ on a standard Borel space is generated by a single measurable function \cite{vonNeumann1930}.} $(\dagger)$ identifies experimental contexts $E\in\cE$ with commutative subalgebras $C_E\subset\cNsa$.

As a consequence, $(\dagger)$ identifies observables in experimental contexts corresponding to the same commutative subalgebra:
\begin{equation}\label{eq: identified observable subalgerbas}
\begin{aligned}
    \cA_E\ \supset\ \cA_{\widetilde{E}}
    \ \stackrel{\dagger}{=}\ 
    C_{\widetilde{E}}=C_{\widetilde{E'}}
    \ \stackrel{\dagger}{=}\ 
    \cA_{\widetilde{E'}}\ \subset\ \cA_{E'}\; .
\end{aligned}
\end{equation}
Crucially, the identifications under Eq.~(\ref{eq: identified observable subalgerbas}) impose additional constraints on (experimental) context connections. It is therefore not clear whether $(\dagger)$ is compatible with Def.~\ref{def: classical recordability} and Prop.~\ref{prop: classical coherence}. Indeed, our main result shows that this is not the case.

\begin{theorem}\label{thm: no observable-operator correspondence}
    The quantum-mechanical observable-operator correspondence \emph{$(\dagger)$},
    \begin{align*}
        \textit{observables}\ \stackrel{1:1}{\longleftrightarrow}\ \textit{self-adjoint operators}\; ,
    \end{align*}
    is incompatible with the notion that quantum observables are consistently recordable.
\end{theorem}

\begin{proof}
    Since $(\dagger)$ is supposed to apply to all quantum systems, it is sufficient to prove inconsistency with recordability in Def.~\ref{def: classical recordability} in the finite-dimensional case.
    
    Let $\cNsa=\BHsa$ with $\dim(\cH)=d$ the finite dimension of the quantum system under consideration, and let $\Em\subset\cE$ be the family of all maximal experimental contexts corresponding to this system, that is, $|\Omega_E|=d$ for every $E\in\Em$.

    We now impose the observable--operator correspondence $(\dagger)$. For every experimental context $E\in\Em$, let $C_E\subset\cNsa$ denote its representative commutative subalgebra under $(\dagger)$, and for every (sharp) observable $X\in\cA_E$, let $a_X\in\cNsa$ denote its representative self-adjoint operator under $(\dagger)$. Moreover, for every event $A\in\Sigma_E$, let $\chi_A\in\cA_E=\cM_b(\Omega_E,\Sigma_E)$ denote its indicator function. Compatibility of $(\dagger)$ with the functional calculus induces a map
    \begin{align*}
        \iota_E:\Sigma_E\longrightarrow\cP(C_E)\; ,\qquad
        \iota_E(A):=(\dagger)(\chi_A)\; .
    \end{align*}
    where $\cP(C_E)$ denotes the projections in the commutative algebra $C_E$. Indeed, since $\chi_A^2=\chi_A$, its representative under $(\dagger)$ is a self-adjoint idempotent, hence, a projection. Now, for every $A=X^{-1}(B)$ with $X\in\cA_E$ and $B\in\cB(\R)$, we have
    \begin{align*}
        \iota_E(A)
        =(\dagger)(\chi_A)
        =(\dagger)(\chi_B\circ X)
        =\chi_B((\dagger)(X))
        =\chi_B(a_X)
        =\pi_{a_X}(B)\; ,
    \end{align*}
    where we used that $\chi_A=\chi_B\circ X$, that $(\dagger)$ is compatible with post-processing by Eq.~(\ref{eq: FUNC}), and wrote $\pi_{a_X}:\cB(\R)\ra\PH$ for the unique projection-valued measure associated with $a_X$ by the spectral theorem. In the finite case, this yields a bijection
    \begin{align*}
        q_E:\Omega_E\lra\cP_1(C_E)\; ,\qquad
        q_E(\omega):=\iota_E(\{\omega\})\; .
    \end{align*}
    
    Let $\nabla$ be an experimental context connection on $\Em$. Since $(\dagger)$ identifies (maximal) experimental contexts with (maximal) commutative subalgebras of $\cNsa$, it follows from Eq.~(\ref{eq: overlap identifications}) that the maps $\widetilde{\nabla}_{C_{E'}C_E}=q_{E'}\circ\nabla_{E'E}\circ q^{-1}_E:\cP_1(C_E)\ra\cP_1(C_{E'})$ satisfy the condition
    \begin{align}\label{eq: overlap condition}
        \left.
        \widetilde{\nabla}_{C_{E'}C_E}
        \right|_{\cP_1(C_E)\cap\cP_1(C_{E'})}
        =\id_{\cP_1(C_E)\cap\cP_1(C_{E'})}\; ,
    \end{align}
    where $\cP_1(C_E)$ denotes rank-$1$ projectors in $C_E$. Hence, $\widetilde{\nabla}=(\widetilde{\nabla})_{C_E,C_{E'}}$ defines a context connection on the maximal commutative subalgebras of $\cNsa$, as defined in Def.~8 of Ref.~\cite{Frembs2025}. Now, if $\cO$ is recordable, there exists a classical coordinatisation $c$ on the set of maximal commutative subalgebras of $\cNsa$, and thus an experimental context connection $\nabla$ that satisfies Eq.~(\ref{eq: flatness}) by Prop.~\ref{prop: classical coherence}. Moreover, this condition is inherited by $\widetilde{\nabla}$ since
    \begin{align*}
        P_\gamma(\widetilde{\nabla})
        =\circ_{i=0}^{n-1}\widetilde{\nabla}_{C_{E_{i+1}}C_{E_i}}
        =\circ_{i=0}^{n-1}\left(q_{E_{i+1}}\circ\nabla_{E_{i+1}E_i}\circ q^{-1}_{E_i}\right)
        =q_{E_0}P_\gamma(\nabla)q^{-1}_{E_0}
        =q_{E_n}q^{-1}_{E_0}
        =\id_{\cP_1(C_0)}\; ,
    \end{align*}
    where we used that context cycles on $\Em$ correspond with context cycles on maximal commutative subalgebras of $\cNsa$, since $E_i\stackrel{(\dagger)}{=}C_{E_i}$, together with $E_n\stackrel{(\dagger)}{=}C_E=C_0\stackrel{(\dagger)}{=}E_n$ and Eq.~(\ref{eq: flatness}). However, by the (reformulation of the) Kochen-Specker theorem \cite{KochenSpecker1967}, Thm.~1 in Ref.~\cite{Frembs2025} (see also Ref.~\cite{Frembs2024}), no such context connections exist for $d\geq 3$.
\end{proof}

Consequently, $\cNsa$ is not a consistently recordable space of observables, as understood by Def.~\ref{def: classical recordability}. And since, as we argued, the latter is intrinsic to the interpretation of $\cNsa$ as an algebra of observables, Thm.~\ref{thm: no observable-operator correspondence} implies that:
\begin{center}
    $(\dagger)$ \textit{is inadequate as a representation of the observable structure of quantum theory.}
\end{center}

In fact, Thm.~\ref{thm: no observable-operator correspondence} asserts that if we `measured' sufficiently many observables\footnote{Note that the problem arises already for finite sets of self-adjoint operators \cite{KochenSpecker1967,Frembs2024,Frembs2025}.} and assume them to be faithfully represented by self-adjoint operators, we could not be certain anymore which operators represent which observables. Indeed, under $(\dagger)$ (and for $d\geq 3$) there \emph{always} exists a tuple of experimental contexts $\gamma=(E_0,...,E_n)$ with $E_0\stackrel{\dagger}{=}C\stackrel{\dagger}{=}E_n$, such that the rank-$1$ projections $(p_1,...,p_d)$ in $C$ are identified with a nontrivial permutation $(p_{\pi(1)},...,p_{\pi(d)})$ with $\pi:=P_\gamma(\nabla)$. But this means that $(\dagger)$ identifies the same nondegenerate observable $X\in E_0\stackrel{\dagger}{=}C$ with different self-adjoint operators,
\begin{align*}
    X\stackrel{\dagger}{\leftrightarrow}a=\sum a_i p_i
    \neq
    \sum a_ip_{\pi(i)}=a'\stackrel{\dagger}{\leftrightarrow}X\; .
\end{align*}

Put another way, $(\dagger)$ thus breaks the eigenvector-eigenvalue link, which is usually understood to relate Hilbert space projections with observable outcomes. Nevertheless, we emphasise that it is \emph{not} sufficient to merely allow the value of an observable to be `contextual', that is, to depend on the experimental context: as long as observables are identified with self-adjoint operators via $(\dagger)$, Thm.~\ref{thm: no observable-operator correspondence} still applies. It is the representation of observables itself, that is, the abstraction from physical observables to their mathematical representation (see Fig.~\ref{fig: parliament}) that needs revision!

\section{Interpreting the result}\label{sec: interpretations}

In this section, we go over some immediate questions that may arise from our claim in Thm.~\ref{thm: no observable-operator correspondence}. We do not have definite answers to all (or even any) of them, but hope to clarify our point by addressing them as best we can, and by starting a discussion.

\subsection{What about self-adjoint operators?}

Thm.~\ref{thm: no observable-operator correspondence} asserts that $\cNsa$ is not consistently recordable. The assumption that the observables of a system are faithfully represented by self-adjoint operators is therefore inconsistent. This throws the standard interpretation of $\cNsa$ as the `algebra of observables' in serious doubt (see also Sec.~\ref{sec: consequences}), since there seems to be no coherent way to map from observables (equivalence classes of operationally indistinguishable experimental setups) to self-adjoint operators (mathematical representatives of observables).

It has to be emphasised that our quibble is not with the mathematical formalism of quantum theory itself. Surely, operator algebras exist as abstract, mathematical objects, and their rich structure theory has undeniably provided much insight into the physics of quantum systems. The question we ask is what justifies their status as (i.e., their identification with) observables? That is, what justifies the correspondence between the experimentally meaningful notion of observables with their mathematical formalisation?

To be clear, nothing in our argument prevents the elements of $\cNsa$ from being labelled as abstract mathematical objects. The obstruction concerns the additional structure inherited from their interpretation as observables via $(\dagger)$. The problem is therefore not that self-adjoint operators cannot be labelled as abstract mathematical objects -- they obviously can -- but that they do not admit a classical coordinatisation, and thus do not admit an interpretation as observables as outlined in Sec.~\ref{sec: observables} and Sec.~\ref{sec: spaces of observables}.

Importantly, the issue is with the `global' interpretation of $\cNsa$ as a space of observables. `Locally', that is, for certain subsets of observables (equivalently experimental contexts), no problem arises, and self-adjoint operators can be employed and interpreted as usual. In fact, our minimal formalisation of observables in Sec.~\ref{sec: observables} led us to commutative von Neumann algebras $\cA_E$. A commutative von Neumann algebra can always be represented on a Hilbert space, with the self-adjoint operators representing the observables in $E$. In this view, the problem relates to the fact that in the standard formalism, all the commutative von Neumann algebras corresponding with different contexts are represented simultaneously on the \emph{same} Hilbert space; clearly, this imposes various identifications between experimental contexts, and obstructs recordability for $\cNsa$.

\subsection{What has gone wrong?}\label{sec: whats wrong}

Before imposing $(\dagger)$, contextual labels distinguish observables via their experimental contexts and Prop.~\ref{prop: classical coherence} is thus unproblematic. Under $(\dagger)$, however, observables in different experimental contexts that correspond to the same operator are declared to represent the same observable. Compatibility with this identification requires the induced context connection to fix the corresponding overlaps of commutative subalgebras. It is precisely this additional condition which turns experimental context connections in Def.~\ref{def: ex context connection} into context connections in Refs.~\cite{Frembs2024,Frembs2025}, which are generally incompatible with Prop.~\ref{prop: classical coherence}.

$(\dagger)$ is usually justified from operational indistinguishability in Eq.~(\ref{eq: operational indistinguishability}). But how does one arrive at this quotient? More precisely, in order to identify (outcomes of) observables via Eq.~(\ref{eq: operational indistinguishability}), it must surely be possible to distinguish observables in the first place, and across different observables prior to their mathematical formalisation? 

Indeed, as pointed out before, this is possible for $\cE_0$, that is, there exists a map
\begin{align}\label{eq: classical coordinatisation}
    c:\bigsqcup_{E\in(\cE_0)_\mathrm{max}}\Omega_E\lra\Lambda\; ,\qquad
    c|_{\Omega_E}=c_E\; ,
\end{align}
where $(\cE_0)_\mathrm{max}\subset\cE_0$ is the family of (pairwise isomorphic) maximal experimental contexts of $\cE_0$. This follows simply since no identifications are being made, that is, $\Omega_E\cap\Omega_{E'}=\emptyset$ in our symbolic notation,$^{\ref{fn: overlap notation}}$ and Eq.~(\ref{eq: event identifications}) thus imposes no constraints on $c$. By Prop.~\ref{prop: classical coherence}, this implies that there exist context connection $\nabla=(\nabla_{E'E})_{E,E'\in(\cE_0)_\mathrm{max}}$ such that Eq.~(\ref{eq: flatness}) holds.
Under $(\dagger)$, maximal experimental contexts are identified with maximal commutative subalgebras of $\cNsa$. Consequently, $(\dagger)$ defines a quotient of the disjoint union of outcome spaces, $q_\dagger:\bigsqcup_{E\in(\cE_0)_\mathrm{max}}\Omega_E\ra\Omega_\dagger:=\left(\bigsqcup_{E\in(\cE_0)_\mathrm{max}}\Omega_E\right)/\sim_\dagger\cong\PHone$, where
\begin{align*}
    (E,\omega)\sim_\dagger(E',\omega')
    \quad\Longleftrightarrow\quad
    q_\dagger|_E(\omega):=q_E(\omega)=q_{E'}(\omega')=q_\dagger|_{E'}(\omega')\; .
\end{align*}
Crucially, under the equivalence relation $\sim_\dagger$ the contextual information disappears. Nevertheless, for the resulting representatives to retain their interpretation as observables with a consistent record, the classical coordinatisation in Eq.~(\ref{eq: classical coordinatisation}) must descend through the identifications of $\sim_\dagger$. More formally, there must exist a map
\begin{align*}
    \overline{c}:\Omega_\dagger\lra\Lambda
\end{align*}
such that $c$ factorises through $q_\dagger$, that is, such that $c=\overline{c}\circ q_\dagger$, as in the following diagram:
\begin{equation*}
\begin{tikzcd}[column sep=large,row sep=small]
    \bigsqcup\limits_{E\in(\cE_0)_\mathrm{max}}\Omega_E
    \arrow[rr,"c"]
    \arrow[dr,"q_\dagger"']
    && \Lambda\\
    &
    \Omega_\dagger
    \arrow[ur,dashed,"\overline{c}"']
    &
\end{tikzcd}
\end{equation*}
Yet, according to Thm.~\ref{thm: no observable-operator correspondence} no such map $\overline{c}$ exists for $\Omega_\dagger$. Consequently, the quotient in Eq.~(\ref{eq: operational indistinguishability}) cannot be performed without losing the interpretation of observables in \ref{ax: EA informal}.

\subsection{Are we asking too much?}

One may suspect that we built a straw man: after all, how important is it for an experimenter to have `simultaneously' a description of all observables? Isn't this unrealistic?  In fact, this is not spurious, and it is not asking too much if, as is usually assumed (such as by $(\dagger)$) that a system admits a kinematical description in terms of states and observables. Kinematics demands a well-defined step from physical entities and procedures to their mathematical representation. Only once this is settled, one can `do physics'. It is here, in this first step of mathematical formalisation, that we take a closer look.

`Simultaneous' does not mean that all experimenters perform measurements at the same time (which is impossible), but rather that they can unambiguously and consistently communicate with each other about their respective experimental contexts and outcomes. This level of mathematical formalisation is required before the results of measurements can be analysed within a physical theory. Instead of several experimenters, it is sufficient to consider a single experimenter who intends to perform a series of measurements. She, too, will have to settle on a mathematical formalisation first.

Def.~\ref{def: classical coordinatisation} should therefore not be understood as an additional physical assumption about experimental contexts. Rather, it expresses the elementary consistency condition that their descriptions can be recorded and communicated in a common classical language, and on this reading is part of the interpretation attached to $(\dagger)$.

Scrutinising this implicit assumption about the mathematical formalisation of observables as self-adjoint operators, we showed that it is internally inconsistent.

\subsection{Is this just Bohr's doctrine of classical concepts?}\label{sec: on Bohr}

\ref{ax: EA informal} and its formalisation via Def.~\ref{def: classical coordinatisation} and Prop.~\ref{prop: classical coherence} are clearly reminiscent of Bohr's doctrine of classical concepts. Throughout his life, he fervently insisted that \cite[pg. 209]{Bohr1949}
\begin{center}
    \textit{[...] the account of all evidence must be expressed in classical terms. The argument is simply that by the word ‘experiment’ we refer to a situation where we can tell to others what we have done and what we have learned and that, therefore, the account of the experimental arrangement and of the results of the observations must be expressed in unambiguous language with suitable application of the terminology of classical physics.}
\end{center}
We agree with this emphasis on the necessity to be able to communicate experimental contexts. Our premise \ref{ax: EA informal} may thus be seen as a condensed version of Bohr's conceptual demand, and Def.~\ref{def: classical coordinatisation} and Prop.~\ref{prop: classical coherence} as a more concrete formalisation of his idea.

Furthermore, our emphasis on sharply separating the formal apparatus of a physical theory from its reference to experimentally accessible phenomena also seems to be closer in spirit to Bohr than to Heisenberg. Whereas Heisenberg sought to reconcile the felt tension between the quantum formalism and its classical description in terms of uncertainty relations, Bohr saw the latter as a consequence of a deeper metaphysical principle, which he coined `complementarity'. Echoing Bohr's sentiment, we remain skeptical whether the retreat behind the walls of probabilistic uncertainty can be deemed sufficient to reconcile the paradox of the mind-boggling accuracy with which quantum-theoretical predictions have been experimentally verified over the last 100 years on the one side, and the highly unsatisfactory status of its axiomatisation on the other, which includes -- as a central piece -- the observable-operator correspondence $(\dagger)$. In Sec.~\ref{sec: consequences}, we point towards directions that may improve this situation.

Before doing so, we point out another parallel between Bohr's doctrine of classical concepts and our result, namely the shared connection with the topos approach to quantum theory \cite{IshamButterfieldI,IshamButterfieldII,IshamButterfieldIII,IshamButterfieldIV,IshamDoeringI,IshamDoeringII,IshamDoeringIII,IshamDoeringIV,HLSBohrification2009,HLS2009,HLS2010,DoeringFrembs2019a}. The latter decomposes the noncommutative algebra of a quantum system into its commutative subalgebras, which are loosely meant to represent experimental contexts. However, as Thm.~\ref{thm: no observable-operator correspondence} shows, this falls short of implementing Bohr's ideas fully. To say it with Bohr \cite[pg 209]{Bohr1949}: while the commutative subalgebras of a noncommutative algebra qualify as `classical terms', when taken together, they are no longer `unambiguous' in describing the `account of all evidence'.

The reason is that, by Thm.~\ref{thm: no observable-operator correspondence}, which uses the reformulation of the Kochen-Specker theorem in Thm.~1 of Ref.~\cite{Frembs2025}, commutative subalgebras cannot consistently describe the experimental contexts they are supposed to represent. Our result therefore not only reaffirms the intimate connection between Bohr's doctrine of classical concepts and no-go theorems in quantum foundations; but it provides a new, and arguably more fundamental reading of the latter. We explain this in more detail in the next section.

\subsection{The Kochen-Specker theorem revisited}\label{sec: KS theorem revisited}

From a technical point of view, our Thm.~\ref{thm: no observable-operator correspondence}, rests on Thm.~1 in Ref.~\cite{Frembs2025}, which constitutes a reformulation of Kochen-Specker contextuality (see also Ref.~\cite{Frembs2024,Frembs2026_ContextualitynqubitPauliGroup}). Indeed, context connections were first defined for the partially ordered set of commutative subalgebras of $\cNsa$, as a reformulation of the Kochen-Specker theorem \cite{Frembs2024,Frembs2025}. Traditionally, the Kochen-Specker theorem \cite{Specker1960,KochenSpecker1967} is interpreted as asserting the non-existence of microstates that would assign definite values to all quantum observables -- assumed to be represented by self-adjoint operators --, as is the case for points in a classical phase space. Consequently, any naive phase-space picture for quantum mechanics fails. While technically we are using the same result (the violation of Eq.~(\ref{eq: flatness})),\footnote{Indeed, the Kochen-Specker theorem may be read as a mere structural result about the noncommutative algebra $\cNsa$.\label{fn: structural KS}} conceptually, our view is reversed -- from the non-existence of microstates on $\cNsa$ to the status of the observable algebra $\cNsa$ itself. It is worth elaborating on this point in more detail.

Noncontextuality relates to the question whether quantum theory admits a classical description.\footnote{The corresponding research programme is often associated with Einstein, and thought to have been initiated through the claim that quantum theory is an `incomplete' physical theory \cite{EPR1935}.} If the answer was yes, such a description could play a similar role for quantum mechanics as statistical mechanics did for thermodynamics: providing a microscopic explanation for macroscopic thermodynamic phenomena. To address this question, Kochen and Specker define noncontextuality as compatibility with the Borel functional calculus, arguably a minimal condition for a map $\epsilon:\cNsa\ra L^\infty(\Sigma,\sigma,\mu)$ to qualify as a representation of self-adjoint operators. The non-existence of such a map is the homonymous theorem in quantum foundations,\footnote{We remark that the Kochen-Specker theorem and contextuality have been generalised in various ways, including \cite{IshamButterfieldI,IshamDoeringI,Doering2004,Spekkens2005,HLSBohrification2009,AbramskyBrandenburger2011,VanDenBergHeunen2012,FritzEtAl2013,CabelloSeveriniWinter2014}. However, all of these are (predominantly) concerned with the mathematical (representational) aspect of the theorem only and, as a result, our criticism above applies essentially in the same way also to those extended frameworks.} which is usually summarised in the slogan: quantum theory is contextual, or more poignantly \cite{Peres1978}:
\begin{center}
    \textit{Unperformed experiments have no results.}
\end{center}
Yet, Peres, Kochen and Specker, as well as many authors refining the proof since then, take as a given that quantum observables are faithfully represented by self-adjoint operators to begin with. Our reading of the Kochen-Specker theorem via Thm.~\ref{thm: no observable-operator correspondence} is different:
\begin{center}
    \textit{Self-adjoint operators are not consistently recordable.}
\end{center}
In this reading, the Kochen-Specker theorem says that while it is surely possible to represent a limited number of observables by self-adjoint operators, there is no coherent way of doing so for the full algebra $\cNsa$. Again, if we assume that sufficiently many observables are faithfully represented by the self-adjoint operators (such that the projections onto their eigenspaces contain e.g. a Kochen-Specker configuration), we could not be certain anymore which operators represent which observables. Certainly, this throws into question the validity of the formalisation of observables in the first place.

This change of perspective is not merely linguistic. The issue with the classically presumed existence of outcomes to counterfactually performed experiments involves the measurement process, which is plagued with numerous unresolved conceptual issues itself, and brings in the inherently probabilistic character of the Born rule. By contrast, the observable-operator correspondence $(\dagger)$ does not depend on states or statistical predictions, if it is supposed to formalise the kinematical structure of quantum theory.

What is more, the above change of perspective on the Kochen-Specker theorem consolidates two aspects of contextuality that remain contentious in the usual reading.

The first is state-independence. Recall that Peres' paper \cite{Peres1978} actually refers to Bell experiments and not the Kochen-Specker theorem. Indeed, the no-go theorems by Bell, Kochen and Specker are often presented together and in a manner that suggests the latter is a generalisation of the former. However, while violations of Bell inequalities are inherently state-dependent, the Kochen-Specker theorem admits state-independent proofs \cite{KochenSpecker1967,KleinmannEtAl2012,Frembs2024,Frembs2025}. This may seem surprising when Kochen-Specker contextuality is viewed as a property adhering to states (and compared to Bell's theorem), yet it is immediate when read as a property of the `algebra of observables' itself.

The second is determinism. It is sometimes pointed out that the notion of (Kochen-Specker) noncontextuality is overly classical in demanding that operators be assigned outcomes in a deterministic way (insinuating that a probabilistic assignment via the Born rule may resolve the issue).\footnote{It is true that quantum states assign probabilities to projections in a way that is independent of context (commutative subalgebras). However, this does not resolve the issue of relating projections with observable outcomes via $(\dagger)$.} Instead, when read as a statement about the observable-operator correspondence $(\dagger)$, determinism simply becomes part of the idea that an algebra of observables exists as a kinematical, operationally accessible notion.

In summary, we suggest that (Kochen-Specker) contextuality is better interpreted as
\begin{itemize}
    \item an algebraic notion, as opposed to a notion referring to (micro)states, and as
    \item relating to the observable-operator correspondence, rather than hypothetical representations $\epsilon$ of the algebra of self-adjoint operators as measurable functions.
\end{itemize}

To emphasise the point, recall the common reading of the Kochen-Specker no-go theorem: any classical (hidden-variable) model of quantum theory is necessarily contextual in that it does not respect the identifications between experimental contexts. This statement is about (noncontextual) classical representations of the quantum formalism (motivated by the analogy with thermodynamics and statistical mechanics). As a statement about the mathematical apparatus of quantum theory this is fine;$^{\ref{fn: structural KS}}$ however, this reading somewhat hides the physical motivation behind the result in the definition of Kochen-Specker noncontextuality, which is based on the interpretation of self-adjoint operators as observables, and thus on the observable-operator correspondence $(\dagger)$.

With Thm.~\ref{thm: no observable-operator correspondence}, what we propose is to instead put this physical assumption at the center, and to read the obstruction of the Kochen-Specker theorem `backwards', namely as a constraint on the mathematical formalisation of physical observables.

\subsection{Why was this not seen before?}

By asking this question we do not mean to claim absolute novelty of the issue we raise. As discussed in the previous sections, there are close parallels with similar ideas dating back to even the earliest days of quantum theory (see Sec.~\ref{sec: on Bohr}). Nevertheless, we believe that our result adds a new dimension and sharpness to the problem (see Sec.~\ref{sec: KS theorem revisited}), which -- in light of the unparalleled success of quantum mechanics -- begs the question how the observable axiom $(\dagger)$ has survived essentially unchanged over the last 100 years.

We offer a few clues towards answering this latter question. The first is that a violation of Eq.~(\ref{eq: flatness}) is hard to measure or even notice experimentally. In fact, unlike a Bell correlation, it is unclear how to measure contextuality in principle -- precisely because the identification of observables across contexts cannot be established without theoretical preconceptions \cite{Winter2014}. Even with added assumptions, violating a `noncontextuality inequality' \cite{Klyachko2008,CabelloEtAl2008,KirchmairEtAl2009,GuhneEtAl2010} requires many, and precisely chosen experimental contexts, which may not be present in most experimental setups (that are not specifically chosen to this way); moreover, as with the violation of Bell inequalities, these are based on states and thus probabilistic in nature, which further incurs issues surrounding the finite precision of real-world measurements \cite{Meyer1999,Kent1999,BarrettKent2004}.

Yet, we believe the real reason to lie elsewhere. As pointed out before, the issue with $(\dagger)$ comes from the fact that the formalisation of observables as self-adjoint operators imposes a huge amount of identifications (of eigenprojections across contexts). In fact, in his independent proof\footnote{Bell's argument against noncontextual hidden variable models is based on Gleasons' theorem \cite{Gleason1957}, and although published later, forms the motivation for his more famous 1964 paper \cite{Bell1964}.} that quantum theory is (Kochen-Specker) contextual, Bell refutes the assumption of noncontextuality, insisting that it still presumes too much of the peculiar features of the quantum formalism. In terms of our analysis above, what he criticises is precisely the identification of experimental contexts with commutative subalgebras of $\cNsa$ via $(\dagger)$. He even calls this identification `unreasonable' \cite{Bell1966}, insisting that the actual implementation of the observables corresponding to different experimental contexts factually involves very different experimental setups, even if these observables are identified under the observable-operator correspondence $(\dagger)$.\footnote{It is interesting, on this account, to remember Bell's later indignation with the term `measurement' or `observable' altogether \cite{Bell1990}. He fervently rejects Heisenberg's proposal to reduce a physical theory to (or at least fundamentally build it from) observable outcomes, and is driven towards a more realist picture, grounded in spatiotemporal and causal primitives. He seems to find this in Bohm's account, which, interestingly, also abandons the observable algebra as a primitive object \cite{DaumerDürrGoldsteinZhangi1996}.}

This resonates with our view: in real experiments, experimental contexts and relations between them are always present in the background, even if not emphasised; even if seemingly redundant `for all practical purposes' \cite{Bell1990}, our result strengthens the call for a more careful analysis of the confounding simultaneity of the unparalleled predictive power of quantum theory with the apparent(?) inconsistency with the physical objects it describes. We end by sketching what we believe to be prime instances of this tension.

\subsection{Pondering the consequences}\label{sec: consequences}

With one of quantum theory's most basic pillars challenged, Thm.~\ref{thm: no observable-operator correspondence} puts into serious distress the interpretation of the noncommutative algebra in the formalisation of quantum theory. We suggested that adding additional contextual data may for most `practical purposes' save the usage of the formalism; still, Thm.~\ref{thm: no observable-operator correspondence} calls into question the physical meaning of noncommutative algebraic operations. There are various aspects to this, we only consider two that seem immediately relevant. The first addresses the noncommutative, associative product of operators in $\cN$, which admits the decomposition
\begin{align*}
    ab=\frac{1}{2}[a,b]+\frac{1}{2}\{a,b\}\; ,
\end{align*}
where $[\cdot,\cdot]$ denotes the commutator and $\{\cdot,\cdot\}$ the anti-commtutator. While the commutator of skew-Hermitian matrices obtains a dynamical interpretation (in terms of symmetry generators) analogous to the Poisson bracket in classical mechanics, the non-associative Jordan algebra of observables $(\cNsa,\{\cdot,\cdot\})$ still escapes a clear physical interpretation. Our result puts yet more emphasis on resolving this problem.

Another aspect relating to the algebraic structure of $\cNsa$ is the status of contextuality as a resource in quantum computation \cite{HowardEtAl2014,DelfosseEtAl2015,Raussendorf2019,FrembsRobertsBartlett2018}. As usually read, contextuality is a property of the mathematical formalism; in turn, we have argued that sufficient specification of experimental contexts effectively renders experimentally accessible spaces of observables noncontextual. One feels prompted to ask: does this put into jeopardy one of the key phenomena suggested to be responsible for quantum advantage?

Another question that relates to the algebraic structure of $\cNsa$ is to what extent the probabilistic nature of quantum theory conceals the problematic status of the observable axiom $(\dagger)$. Do Heisenberg's uncertainty relations imply sufficient indeterminacy between incompatible measurements that their impossibility to be labelled across many experimental contexts becomes irrelevant? In turn, can probabilities be understood as arising from the contextual nature of self-adjoint operators? And is this a property that is special to quantum mechanics? This question is pertinent both for the topos approach to quantum theory, as well as the general reconstruction programme \cite{Hardy2001,CliftonBubHalvorson2003,Grinbaum2007,ChiribellaDArianoPerinotti2011,MasanesMuller2011,Hoehn2017,Muller2021}.

Finally, we echo another idea of Bohr's: that relativity and complementarity (or maybe better, contextuality) both express different aspects of the same underlying theme: the necessary specification of referential data -- be it in the form of spatiotemporal reference frames, or the specification of experimental contexts -- as part of observational input to theory construction. Presently, this qualitative idea is lacking a unifying mathematical framework, however, we interpret the largely independent research programmes of topos quantum theory, quantum reference frames \cite{LoveridgeMiyaderaBusch2018,GiacominiCastroRuizBrukner2019,CastroRuizOreshkov2025,DeVuystHohnTsobanjan2026} and relational approaches to quantum theory \cite{Rovelli1996,FrauchigerRenner2018,Brukner2018,Bong2020} to all be concerned with this underlying question.\\

\textbf{Acknowledgments.} AD thanks Verena and Julian as well as Chris Isham for continuous inspiration.

\clearpage

\bibliographystyle{abbrv}
\bibliography{bibliography}

\end{document}